\documentclass[aps, notitlepage, onecolumn, pra]{revtex4-1}

\usepackage{graphicx, xcolor}
\usepackage{dcolumn,lipsum}
\usepackage{bm}
\usepackage{amssymb,amsmath,enumerate,epsfig}
\usepackage{pdfpages}

\newcommand{\be}{\begin{equation}}
\newcommand{\ee}{\end{equation}}
\newcommand{\ben}{\begin{enumerate}}
\newcommand{\een}{\end{enumerate}}
\def \bea{\begin{align}}
\def \eea{\end{align}}
\newcommand{\half}{\frac{1}{2}}

\newcommand{\ra}{\rangle}

\newcommand{\braket}[2]{\langle #1|#2\rangle}

\newcommand{\ketbra}[2]{|#1\rangle\!\langle#2|}
\newcommand{\ket}[1]{|#1\rangle}

\makeatletter

\newcommand{\Rmnum}[1]{\expandafter\@slowromancap\romannumeral #1@}
\makeatother

\newenvironment{proof}{\noindent \textbf{{Proof~} }}{\qed}

\newtheorem{thm}{Theorem}

\newtheorem{prop}[thm]{Proposition}
\newtheorem{cor}[thm]{Corollary}
\newtheorem{defn}[thm]{Definition}
\newtheorem{examp}[thm]{Example}

\newtheorem{obs}[thm]{Observation}

\def\li{\left}
\def\pr{\right}
\providecommand{\Q}{{\cal{Q}}}
\providecommand{\PP}{{\cal{P}}}

\def\ket#1{| #1\rangle}

\def\ra{\rangle}

\def\>{\rangle}
\def\<{\langle}

\def\squareforqed{\hbox{\rlap{$\sqcap$}$\sqcup$}}
\def\qed{\ifmmode\squareforqed\else{\unskip\nobreak\hfil
\penalty50\hskip1em\null\nobreak\hfil\squareforqed
\parfillskip=0pt\finalhyphendemerits=0\endgraf}\fi}

\newcommand{\DAMTP}{Department of Applied Mathematics and Theoretical Physics, University of Cambridge, Cambridge CB3 0WA, U.K.}

\begin{document}

\input epsf
\title{
Optimal amount of entanglement to distinguish quantum states
  instantaneously}
\author{Berry Groisman}
\affiliation{\DAMTP}
\author{Sergii Strelchuk}
\affiliation{\DAMTP}

\begin{abstract}
We introduce a new aspect of nonlocality which arises when the task of quantum states distinguishability is considered under local operations and shared entanglement in the absence of classical communication. We find the optimal amount of entanglement required to accomplish the task perfectly for sets of orthogonal states and argue that it quantifies information nonlocality.

\end{abstract}

\maketitle
\section{introduction}
According to the postulates of quantum mechanics any set of orthogonal states of a {\it single} quantum system in principle can be perfectly (i.e. with probability $1$) distinguished by an appropriate projective measurement. A composite system can be treated effectively as a single system if one can implement a global projection measurement on it.

It is interesting to look at nontrivial scenarios, where the class of possible global projection measurements is restricted. Such scenarios arise naturally in quantum information theory when two or more quantum systems are separated in space, in which case implementation of a global operation on these systems will be restricted by the available resources of communication (quantum and classical). The implications of these restrictions on distinguishability properties of sets of orthogonal joint \footnote{Here {\it joint} does not necessary mean {\it entangled}} states have been used extensively as a tool to study fundamental properties of quantum systems~\cite{chefles_condition_2004, de_rinaldis_distinguishability_2004, walgate_local_2000, chefles_unambiguous_1998}.

Communication resources in quantum information theory can be divided into classical and quantum. The abundance of both types of resources trivialises the task of quantum state distinguishability, since all the systems can be teleported to one location where a required global operation can be implemented, after which the systems are teleported back to their original locations. Restrictions on available resources affect distinguishability and particularly interesting operational settings arise when one type of resource is completely prohibited. It is under these restrictions that the nature of quantumness of the states is revealed. The first case corresponds to local operations, accompanied by classical communication, without access to shared entanglement (LOCC). The second case is a dual to LOCC --- where we are restricted to local operations and shared entanglement (LOSE).

It is natural to expect that the nonlocal nature of composite systems with parts separated in space will manifest itself in their distinguishability properties. In the past it was widely believed that entangled states are the only states that exhibit nonlocal behaviour.
However, it was shown that there are aspects of quantum states which, albeit not being associated with entanglement, are difficult to reconcile with locality. It might not come as a surprise that some sets of orthogonal entangled states cannot be distinguished by LOCC. What does come as a surprise is that certain sets of mutually orthogonal {\it product} (i.e. nonentangled) multipartite states cannot be perfectly distinguished by LOCC. This is contrary to what one would have expected from a set of classical states. Thus,
despite the absence of entanglement in the states themselves there is a certain aspect of nonclassicality (quantumness) to such states.
The fact that their distinguishability cannot be facilitated without utilizing entangled resources prompted researchers to conclude that there is a new type of nonlocality involved, {\it nonlocality without entanglement}~\cite{NLWI}.

The change of paradigm from LOCC to LOSE has profound implications on distinguishability. There exist families of sets which are perfectly distinguishable by LOCC with communicating as little as one classical bit, yet, cannot be perfectly distinguished by LOSE with arbitrary large but finite resources. As LOSE does not involve classical communication, corresponding protocols can be considered as instantaneous. The parties can implement their local interactions --- unitary operations and measurements --- in a very short time: that is, these operations occur in spacelike separated regions. During this process all the (classical) results are recorded locally. The parties still need to combine their local classical records in order to identify the initial state. However, the latter step is not considered to be a part of the measurement process~\cite{aharonov_states_1980}.

Distinguishability under LOCC is well understood~\cite{chefles_condition_2004, de_rinaldis_distinguishability_2004, walgate_local_2000, chefles_unambiguous_1998}. However, little is known about distinguishability under LOSE with limited resources. It was shown that for finite-dimensional systems any multipartite mutually orthogonal set can be discriminated with unlimited resources of entanglement ~\cite{IMNLV_GR, IMNLV_vaidman}. The latter works did not address the question of finding optimal amounts of entanglement required to distinguish particular sets.
The characteristic feature of the above protocols is that they consume all pre-shared entanglement. Recently, an alternative scheme had been proposed ~\cite{clark_entanglement_2010} where a finite number of ebits is consumed {\it on average}. The protocol is designed to halt as soon as all required local operations are completed successfully leaving the rest of the pre-shared entanglement resources intact. It should be noted, however, that with nonzero probability such protocols consume entanglement well beyond their expectation value. As such, their worst-case performance is still very far from the optimal.

In this paper we find the minimum amount of resources which must be {\it available} to the parties beforehand in order to accomplish the task perfectly. This leads us to introducing the new aspect of nonlocality. We show that the algorithm of remote instantaneous rotations introduced in~\cite{IMNLV_GR} requires an optimal amount of entanglement with respect to this aspect.  The method of instantaneous teleportations suggested in \cite{IMNLV_vaidman} will consume the same amount of entanglement.

Several recent works endeavour to construct efficient protocols and to find better bounds on the amount of entanglement required~\cite{clark_entanglement_2010, beigi_simplified_2011}. One important application of the improved bounds for entanglement consumption is attacks for position-based cryptography. They rely on the ability to perform state transformations efficiently. It is currently known that an adversary having access to an exponential amount of entanglement can successfully attack a number of position-based verification protocols. It is currently not known whether one can do better. Our work provides additional evidence that only a polynomial amount of entanglement may be required to perform successful attacks on position-based cryptography protocols~\cite{buhrman_position-based_2014}.

This article is organised as follows. In Section II we formulate the problem in the form of a steering game and define information nonlocality. Section III presents a general steering algorithm. Section IV provides a proof of optimality of that algorithm in terms of entanglement resources. In Section V we discuss generalizations of the steering sets.

\section{Steering game and new measure of information nonlocality}
Consider a bipartite quantum state drawn from the known ensemble $\{\ket{\psi_i},p_i\}_{i=1}^m$  of mutually orthogonal bi-partite states shared by two parties. Their uncertainty about the state is quantified by its entropy.

In the absence of classical communication, the manner in which the information about the state can be
accessed locally dictates the nonlocality property of the set.
We say that information encoded in the ensemble (set) is {\it bi-localised} if it can be accessed by the parties using local operations (here accessed does not imply interpreted cf.~\cite{aharonov_states_1980})
and {\it de-localised} if it is not accessible by local operations alone. In the latter case we say that the ensemble is associated with {\it information nonlocality}. The aim of the parties in our setup is to bi-localise the nonlocal information encoded in the set of quantum states by means of LOSE. Thus, we will refer to it as
{\it information nonlocality of instantaneous bi-localisation}.
It is ultimately linked with perfect distinguishability, which operationally depends on parameter(s) characterising the states $\{\ket{\psi_i}\}$, but not on $\{p_i\}$. Hence we expect that a quantitative measure of information nonlocality should not depend on the latter. Thus, we consider the task of state discrimination in the presence of de-localised information. The optimal amount of resources required to complete the task is quantified by the number of ebits needed to access, i.e., bi-localise, the de-localised information. Clearly, LOSE is an appropriate framework to analyse the bi-localization of classical information because it prohibits the classical communication between two parties which would otherwise trivialise the task.

We now introduce a {\it steering game} and a new quantity which will serve as a figure of merit for this game. The latter quantifies the information nonlocality in sets of orthogonal quantum states.
To illustrate this task,  consider two parties, Alice and Bob, who share a state from the known set
\be \label{twisted_product}{\cal A}[\alpha]=\{|0\ra_A|0\ra_B,|0\ra_A|1\ra_B,|1\ra_A|+_\alpha\ra_B,|1\ra_A|-_\alpha\ra_B\},
\ee
where $|+_\alpha\ra=\cos\alpha|0\ra+\sin\alpha|1\ra$, $|-_\alpha\ra=\sin\alpha|0\ra-\cos\alpha|1\ra$.
 When it does not cause confusion, we will omit the subscripts which denote the label of the subsystem. The state is chosen and
prepared by an external party and is unknown to Alice and Bob. By means of LOSE parties try to steer set $\cal A$ to set ${\cal B}=\{|0\ra|0\ra,|0\ra|1\ra,|1\ra|0\ra,|1\ra|1\ra\}$ (with no particular ordering of states) minimizing the amount of entanglement used in the process. We denote the
task of converting the states from set ${\cal A}$ to set ${\cal B}$ with probability 1 as ${\cal A}\rightarrow {\cal B}$.
 For both sets, the overall phase in front of the states will play no role and will further be ignored.

\begin{defn}
{\it We say that the set of two-qubit states ${\cal A}[\alpha]$ contains ${\cal I}({\cal A}[\alpha])= k$ bits of nonlocality of instantaneous bi-localisation, where
${\cal I}({\cal A}[\alpha])$ is the length of the binary expansion of $\alpha \bmod \frac{\pi}{2}$:}

\be
{\cal I}({\cal A}[\alpha]):=|(\bar\alpha)_2|.
\ee
\end{defn}

\noindent It was demonstrated in~\cite{IMNLV_GR} how this game is won using ${\cal I}({\cal A}[\alpha])$ pre-shared ebits. In Section~\ref{optimalitysection} we prove that this number of ebits is indeed optimal. Thus, we treat ${\cal I}$ as a measure, quantifying the optimal amount of resources to steer the set $\cal A$ defined in~\eqref{twisted_product} to the set $\cal B$.\\

\noindent Hence, for finite length of the binary expansion of $\alpha \bmod {\frac{\pi}{2}}$ one can steer perfectly with a finite amount of entanglement~\cite{IMNLV_GR}. However, when the expansion of $\alpha \bmod {\frac{\pi}{2}}$ is infinite, then all known
algorithms use an infinite amount of entanglement required to achieve ${\cal A}\rightarrow {\cal B}$~\footnote{In~\cite{clark_entanglement_2010} authors derive finite expectation value for entanglement consumption, but infinite amount must be available.}.\\

The use of entanglement to accomplish the task is essential, although not immediately clear. One cannot perform ${\cal A}\rightarrow{\cal B}$ using LO with a
weaker resource such as common randomness because in the absence of classical communication it can be seen to represent a pre-agreed strategy. As such, common randomness is
equivalent to Alice and Bob committing to a sequence of LO in advance.\\

Our definition is not sensitive to the probability distribution of the four states (\ref{twisted_product}) because we require the probability of success $p=1$. If they are equiprobable, i.e. $p_i=1/4$ $\forall i$,
then they encode two classical bits of which $1-H_2(\cos^2\alpha)$ are de-localised. This quantity, known as zero-way quantum information deficit~\cite{horodecki_local_2005}, is unsuitable to quantify information nonlocality of the set for two main reasons. First, the above deficit depends on the probability distribution $\{p_i\}$ of the corresponding states, which is
irrelevant as far as distinguishability is concerned. Clearly, physical resources needed to manipulate this set of states cannot depend on that probability
distribution. Second, it is natural to expect that the measure is related to uncertainty about the basis on Bob's side in (\ref{twisted_product}). Thus, our operational understanding of information nonlocality of the set $\{\ket{\psi_i},p_i\}$ is quite different from the previously suggested ways of quantifying nonlocality of the mixtures $\rho_{AB}=\sum p_i \ketbra{\psi_i}{\psi_i}$, where $\ket{\psi_i}\in \cal A$. Our measure also differs from accessible information~\cite{dallarno_informational_2011}, which quantifies mutual information between an ensemble and the outcomes of an optimal measurement, and as such depends on the probability distribution. In our distinguishability scenario accessible information is always maximal, i.e. two bits, while the figure of merit is the amount of entanglement required to carry out a measurement, which achieves this value. A similar comparison can be made with relative entropy of quantumness~\cite{groisman_quantumness_2007} and quantum discord~\cite{ollivier_quantum_2001}.

The parametrization of $\cal A$ can be easily extended to multiple angles:
\begin{equation}\label{set1}\begin{split}
{\cal A}[\alpha_1,...,\alpha_n] =\{&|0\ra|0\ra,|0\ra|1\ra,\\&|1\ra|+_{\alpha_1}\ra,|1\ra|-_{\alpha_1}\ra,\\&|2\ra|+_{\alpha_2}\ra,|2\ra|-_{\alpha_2}\ra,\\
&\vdots\\
&|n\ra|+_{\alpha_n}\ra,|n\ra|-_{\alpha_n}\ra\}
\end{split}
\end{equation}
where in general $\alpha_i\neq\alpha_j$ when $i\neq j$. In this case we define ${\cal I}({\cal A}[\alpha_1,..,\alpha_n]) = \max_{1\le i\le n}|(\bar{\alpha}_i)_2|$. Then, the set $\cal B$ becomes ${\cal B}=\{|i\ra|0\ra, |i\ra|1\ra\}_{i=0}^{n}$. In Sec. \ref{genstsect} we will show how the task can be accomplished using ${\cal I}({\cal A}[\alpha_1,..,\alpha_n])$ pre-shared ebits.\\

The above extension yields the following two properties of ${\cal I}$ (here we assume $|(\bar\alpha_i)|_2<\infty$):

\begin{itemize}
\item
${\cal I}$ is $\it faithful$: ${\cal I} ({\cal A}[\alpha_1,...,\alpha_n])=0$ iff $\forall i \mbox{ }\alpha_i=0 $.

\item ${\cal I}$ is $\it subadditive$ with respect to $\alpha_i$:
\be
{\cal I}({\cal A}[\alpha_1,\alpha_2]) \le {\cal I}({\cal A}[\alpha_1])+{\cal I}({\cal A}[\alpha_2]).
\ee
It trivially holds for $\alpha_1=\alpha_2$.
Considering sets ${\cal A}$ with more nonlocal states corresponding to different values of $\alpha_i$ increases the complexity of the steering task. However, as we show in Sec.  \ref{genstsect} the required number of ebits is $\max_i|(\bar{\alpha}_i)_2|$. This is due to the fact that Alice and Bob may take advantage of the collective local transformation in order to reset the corresponding bit in the binary expansion regardless of the phase. The collective local operation of each of the parties will result in resetting the lowest bit in the binary expansion of each of the phases without the knowledge
about the precise identity of the shared state.
\end{itemize}
\section{General steering algorithm}\label{genstsect}
In this section we present the protocol $\Q^{(k)}\li[\alpha_1,...,\alpha_n\pr]$, which steers the set ${\cal A}\li[\alpha_1,...,\alpha_n\pr]$ to $\cal B$ using $k$ ebits, and show that it is optimal with respect to
$\cal I$. Recall that Alice and Bob are given one of the elements from ${\cal A}\li[\alpha_1,...,\alpha_n\pr]$, and their task is by performing local operations
to turn it into one of the states of the set $\cal B$ using the minimal amount of the shared entanglement. We present a protocol for the steering game and in the subsequent section prove its optimality.

\subsection{A strategy for ${\cal A}\li[\alpha_1,...,\alpha_n\pr]\to {\cal B}$}\label{iteration}
The basic building block of the protocol is a controlled instantaneous probabilistic rotation introduced in \cite{IMNLV_GR}. The parties share $k$ ebits,
$\ket{\Phi^+}_{a_t b_t}$ ($t=1,...,k$), where $k=\max_i {\cal I}\li({\cal A}[\alpha_i]\pr)$ (we assume that $k$ is finite).
The algorithm they execute constitutes of the following pre-agreed sequence of actions, which they perform $k$ times. Each iteration $t$ utilises one entangled
pair $\ket{\Phi^+}_{a_t b_t}$ and results in changing the set from ${\cal A}_{t}\li[\alpha_1^{(t)},...,\alpha_n^{(t)}\pr]$ to ${\cal A}_{t+1}\li[\alpha_1^{(t+1)},...,\alpha_n^{(t+1)}\pr]$.
\begin{enumerate}
\item Bob implements a C-NOT interaction between the qubit $b_t$ and the target qubit $B$:
\begin{equation}
U_{b_t B}^{CNOT}=\ketbra{0}{0}_{b_t}\otimes\mathbb{I}_B+\ketbra{1}{1}_{b_t}\otimes\sigma_y^B,
\end{equation}
where index $t$ denotes the iteration round of the protocol.

\item Bob then measures $\sigma_x$ on particle $b_t$ and records the result, $r_{b_t}=\pm1$.

\item Alice performs a controlled rotation
\begin{equation}\label{c_rot}
U_{A a_t}=\sum_{j=0}^n \ketbra{j}{j}_A\otimes R_{a_t}^x(\alpha_j^{(t)}),
\end{equation}
where $R_{a_t}^x(\alpha_j^{(t)}) \equiv\exp\{i\alpha_j^{(t)}\sigma_x^{a_t}\}$, and $\alpha_0=0$, which results in rotating qubit $B$ by $r_{b_t}\alpha_j^{(t)}$, depending on the state $\ket{j}$ of Alice's qubit $A$.

\item If $r_{b_t}=+1$, then the set is transformed to ${\cal A}^{+}_{t+1}\li[0,...,0\pr]=\cal B$, while the case of  $r_{b_t}=-1$ corresponds to ${\cal
    A}^-_{t+1}\li[2\alpha_1^{(t)},...,2\alpha_n^{(t)}\pr]$. In the former case, Bob ceases his actions (Alice continues with her actions, but they do not affect the
    state of qubit $A$). In the latter case the protocol proceeds to the next iteration step.

\end{enumerate}

\noindent Thus, on each step there is a $50\%$ chance of performing the map ${\cal A}_{t}\to {\cal B}$. The algorithm terminates with certainty at $t=k$,
because ${\cal A}^{\pm}_{k+1}=\cal B$. (It is easy to see that when $n=1$ the protocol above reduces to the protocol described in~\cite{IMNLV_GR}.)

The actions of Alice and Bob take place in two spacelike separated regions, therefore the protocol can be presented as a tensor product of two unitary transformations (see Sec.~\ref{sec:instprot}). In addition, the Appendix provides a very elegant presentation of the protocol in the language of state-operators --- {\it stators} --- introduced in~\cite{reznik_remote_2002}. The latter is a universal resource for remote LOCC and LOSE operations.

\subsection{Binary representation}
It is convenient to represent the set ${\cal A}\li[\alpha_1,...,\alpha_n\pr]$ as a table in which row $i$ contains the binary expansion of $\alpha_i$ (modulo $\pi/2$).
Alice and Bob compute $k$,  and construct a table with
binary expansions of $\alpha_i$ aligned by the binary order. When the expansion for $\alpha_i$ is less than $k$ they pad it with zeros until its length is
$k$. The example for the set  ${\cal A}\li[\frac{7\pi}{16},\frac{5\pi}{16},\frac{3\pi}{16},\frac{1\pi}{4}\pr]$ is shown in Table \ref{T1}.

\begin{table}[h!]
\begin{center}
\begin{tabular}{ccc} \hline\noalign{\smallskip}
$i~$&$\frac{\alpha_i}{\pi}$ &$(\frac{\alpha_i}{\pi}\text{mod} \frac{1}{2})_{_2}$ {\smallskip} \\
\hline\hline\noalign{\smallskip}
1~ &7/16 &0.111\\
\hline\noalign{\smallskip}
2~ &5/16 &0.101 \\
\hline\noalign{\smallskip}
3~ & 3/16 &0.011\\
\hline\noalign{\smallskip}
4~ &1/4 &0.100\\
\hline
\end{tabular}
\end{center}
\caption{Binary representation of $\alpha_i$ for ${\cal A}\li[\frac{7\pi}{16},\frac{5\pi}{16},\frac{3\pi}{16},\frac{1\pi}{4}\pr]$. The values are
calculated modulo $1/2$.}\label{T1}
\end{table}
\noindent The ultimate aim of the protocol is to obtain a table with all entries being zero, which corresponds to $\cal B$. There is a simple global algebraic
operation which nullifies all the nonzero entries. Its effect is equivalent to multiplication by $2^k$.

The advantage of the binary representation is to clarify and simplify the action of the protocol. Each iteration step
results in the entries of the table being multiplied either by $2^{k-t}$ (if $r_{b_t}=+1$) or by $2$ (if $r_{b_t}=-1$). Alice and Bob update their tables
accordingly. Bob updates the table depending on the value of $r_{b_t}$. If $r_{b_t}=+1$ then no further action is required on his side. Alice, who does not have access to
$r_{b_t}$, updates her table as if $r_{b_t}=-1$ and carries on to the next iteration, where she implements (\ref{c_rot}) with new values of $\alpha_i^{(t+1)}$
according to her new table. The algorithm is guaranteed to terminate at $t=k$. The example of the protocol for the set given in Table \ref{T1} is schematically shown in Figure 1.

\begin{figure}
\includegraphics[scale=0.59]{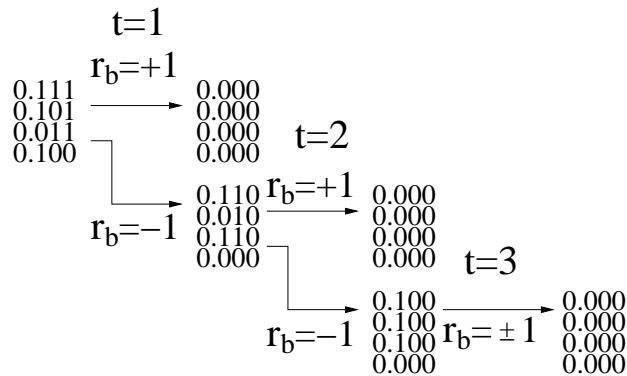}
\epsfxsize=3.2truein
        \caption[]{Schematic representation of the steering protocol for
${\cal A}\li[\frac{7\pi}{16},\frac{5\pi}{16},\frac{3\pi}{16},\frac{1\pi}{4}\pr]$. On each 
$2^{k-t}$ or $2$ depending on
the value of $r_{b_t}$. The resulting values are then taken modulo $1/2$.}
    \label{steering1} \end{figure}
The result of the application of the above protocol resembles that of~\cite{vaidman_qubits_2004} where authors learned each digit of the binary expansion sequentially.

Our protocol can be viewed as one which bi-localises information bit-by-bit and this may be interpreted as the resetting of delocalised information with the cost of resetting quantified by the number of ebits expended, much in the spirit of Landauer's principle.

\subsection{Approximate steering}
In practical applications, we would not seek perfect discrimination of states, i.e. we would not need the perfect steering, namely turning the set ${\cal A}[\alpha]$ to ${\cal B}$. Instead, we may want to obtain a guarantee that we are close to ${\cal B}$ with respect to some notion of proximity defined below. This is especially relevant if there exists a state with a very large length of binary expansion or in the case when $\alpha /\pi$ is an irrational number --- infinite expansion.
\begin{defn}
The set ${\cal A}[\alpha_1,...,\alpha_n]$ is $\epsilon$-close to $\cal B$ if the error
probability of discriminating any pair of the states from $\cal A$ under local operations (with no entanglement and classical communication) is $p_e \le \epsilon$.
\end{defn}

Consider $n=1$ and each of the states in $\cal A[\beta]$ is equiprobable. In this case, the optimal discrimination of states from ${\cal A}[\beta]$ is achieved when Bob simply measures qubit $B$ in the basis $\{\ket{+_{\beta/2}},\ket{-_{\beta/2}}\}$~\cite{groisman_nonlocal_2001}, with probability of error $p_e=\sin^2 (\beta/2)=\epsilon$. Thus, $\forall$ $\alpha\in (0, \pi/4)$ and $\epsilon\in (0,1)$  there exists $\beta$ such that ${\cal A}[\beta]$ is $\epsilon$-close to $\cal B$ and we can steer ${\cal A}[\alpha]$ to ${\cal A}[\beta]$. It is easy to show that for a nonuniform distribution of the states in ${\cal A}[\alpha]$ the error probability $p_e$ is upper bounded by $\epsilon$.

Alice and Bob perform iterations of the algorithm of Sec.~\ref{iteration} which resets the digits starting from the position which corresponds to the position of the most significant digit in the expansion of $2\epsilon$. Hence, the number of ebits which the parties consume is equal to the length of the binary expansion of the portion of $\alpha$ until the position of the most significant digit in the expansion of $2\epsilon$.

\section{Proof of optimality}\label{optimalitysection}
We now prove our main result --- the optimality of the steering algorithm for set ${\cal A}[\alpha]$. It generalises to multiple sets of angles in a straightforward manner. As a preparatory step we introduce several useful definitions and describe generic properties of protocols.
  
\begin{defn}
Consider a set ${\cal A}[\alpha]$ and $k$ ebits $\ket{\Phi^+}_{a_ib_i}$, $i=1,..,k$, shared by Alice and Bob. Assume that there exists a protocol, $\PP^{(k)}[\alpha]$, that is a sequence of pre-agreed local operations performed by Alice and Bob, that achieves  ${\cal A}[\alpha]\rightarrow\cal B$. Without loss of generality, any such protocol can be represented as a two-step process.

\begin{enumerate}[(1)]
\item A tensor product of local unitary operations, $V_{A\tilde {a}_k}\otimes U_{B\tilde {b}_k}$, acting on qubits $A,B$ and all ancillas $a_i,b_i$, which results in transformation

\be
V_{A\tilde {a}_k}\otimes U_{B\tilde {b}_k}\left({\cal A}[\alpha] \bigotimes_{i=1}^{k}\ket{\Phi^+}_{a_i b_i}\right)={\cal C}^{(k)}[\alpha],
\ee
where ${\cal C}^{(k)}[\alpha]$ is the resulting set of entangled states of all particles and $\tilde {a}_k = a_1\ldots a_k$, $\tilde {b}_k = b_1\ldots b_k$.

\item Projection measurements on all ancillas $a_i$, $b_i$ and the original qubits $A$, $B$ in local computational bases.
\end{enumerate}
Henceforth we will only consider the first step of the protocol.

\end{defn}

\begin{defn}
Two sets of bipartite states ${\cal C}^{(k)}[\alpha]$ and ${\cal C}^{(k)}[\beta]$ are said to be {\it locally equivalent} (LE) if they can be interconverted via local operations only. Without loss of generality any such local operation can be assumed to be unitary.
\end{defn}

It is useful to introduce the following notation for the original set supplemented by $k$ resource ebits:
\be
{\cal A}^{(k)}[\alpha]\equiv{\cal A}[\alpha] \bigotimes_{i=1}^{k}\ket{\Phi^+}_{a_i b_i}.
\ee

\noindent Consider two sets ${\cal A}^{(k)}[\alpha]$ and ${\cal A}^{(k)}[\beta]$ and assume that there exist protocols, $\PP^{(k)}[\alpha]$ and $\PP^{(k)}[\beta]$, which steer ${\cal A}[\alpha]$ and ${\cal A}[\beta]$ to ${\cal C}^{(k)}[\alpha]$ and ${\cal C}^{(k)}[\beta]$ respectively. Sets ${\cal A}^{(k)}[\alpha]$ and ${\cal C}^{(k)}[\alpha]$ (and similarly, ${\cal A}^{(k)}[\beta]$ and ${\cal C}^{(k)}[\beta]$) are locally equivalent by virtue of existence of $\PP^{(k)}[\alpha]$ and $\PP^{(k)}[\beta]$.

\begin{prop}\label{master_set}
Assume there exist two distinct protocols $\PP^{(k)}[\alpha]$ and ${\tilde \PP}^{(k)}[\alpha]$, which steer ${\cal A}^{(k)}[\alpha]$ to ${\cal C}^{(k)}[\alpha]$ and ${\cal{\tilde C}}^{(k)}[\alpha]$, respectively. Then, ${\cal C}^{(k)}[\alpha]$ and ${\cal{\tilde C}}^{(k)}[\alpha]$ are locally equivalent.
\end{prop}
\begin{proof}
Observe that 
\be
{\cal C}^{(k)}[\alpha]\rightarrow_{(\PP^{(k)} [\alpha])^{-1}}{\cal A}^{(k)}[\alpha]\rightarrow_{\tilde{\PP}^{(k)}[\alpha]}{\cal{\tilde C}}^{(k)}[\alpha], 
\ee
i.e. ${\cal C}^{(k)}[\alpha]$ and ${\cal{\tilde C}}^{(k)}[\alpha]$ are related by the unitary 
$\tilde{\PP}^{(k)}[\alpha](\PP^{(k)} [\alpha])^{-1}$.

\end{proof}
Observation \eqref{master_set} will be instrumental for our further analysis as it will allow us to restrict the discussion to a single set.
\begin{obs}
Consider two orthogonal subspaces spanned by the first two states in ${\cal A}[\alpha]$, ${\cal A}_0[\alpha]=\{\ket{00},\ket{01}\}$ and the last two states, ${\cal A}_1[\alpha]=\{\ket{1+_{\alpha}},\ket{1-_{\alpha}}\}$. An optimal protocol $\PP^{(k)}[\alpha] =\PP^{(k)}_1[\alpha]\oplus \PP^{(k)}_2[\alpha]$ can be defined by its action on the respective subspaces, i.e.

\be\begin{cases}\begin{split}
{\cal A}_0^{(k)}[\alpha]&\rightarrow_{\PP^{(k)}_0[\alpha]}{\cal C}_0^{(k)}[\alpha],\\
{\cal A}_1^{(k)}[\alpha]&\rightarrow_{\PP^{(k)}_1[\alpha]}{\cal C}_1^{(k)}[\alpha].
\end{split}\end{cases}\ee
${\cal A}_0[\alpha]$ and ${\cal A}_1[\alpha]$ are locally equivalent. [Indeed, they are related by a transformation $\mathbb{I}_A\otimes R_{B}^y(\alpha)$.]
\end{obs}

\subsection{The instantaneous rotations protocol}\label{sec:instprot}

Earlier, in Section~\ref{genstsect}, we described the step-by-step protocol $\Q^{(k)}[\alpha]$ based on instantaneous rotations, which steers ${\cal A}^{(k)}[\alpha]$ to ${\cal C}^{(k)}[\alpha]$ for $\alpha=\pi/2^{k+1}$. This protocol can be written as the product of local unitaries of the form
\be\label{instrot}
\Q^{(k)}[\alpha]=\left(\prod_{i=1}^k\Lambda_{Aa_i}(2^{i-1}\alpha)\right)\otimes\left(\left[\prod_{i=2}^k\left((1-P_1^{(i-1)})\otimes \mathbb{I}_{Bb_i}+P_1^{(i-1)}\otimes\Omega_{Bb_i}\right)\right ]\Omega_{Bb_1}\right),
\ee

\noindent where
\begin{equation}\begin{split}\label{unitaryprot}
\Lambda_{Aa}(\theta)&=\ketbra{0}{0}_A\otimes \mathbb{I}_a+\ketbra{1}{1}_A\otimes R_a^x(\theta),\\
\Omega_{Bb}&= \mathbb{I}_B\otimes\ketbra{+}{0}_b+\sigma_y^B\otimes\ketbra{-}{1}_b,\\
P_1^{(j)}&=\ketbra{1_{b_1}1_{b_2}...1_{b_j}}{1_{b_1}1_{b_2}...1_{b_j}},
\end{split}
\end{equation}
with $\ket{\pm}=(\ket{0}\pm\ket{1})/\sqrt{2}$. 
The task can be accomplished with $m>k$ ebits by acting on the remaining $m-k$ ebits by $Z = \mathbb{I}_{a_j}\otimes\Omega_{Bb_j}$. We can compare any protocol that requires $k$ ebits with the one that requires $m>k$ ebits by padding the former with the $m-k$ ebits and acting on them by $Z$. This will result in the same number of terms in respective sets ${\cal C}$.\\

\noindent As a simple example, consider the action of $\Q^{(1)}[\pi/4]$=$\Lambda_{Aa}(\pi/4)\otimes \Omega_{Bb}$, which results in ${\cal C}^{(1)}[\pi/4]$:

\begin{equation}\label{statesplitting}\begin{split}
\ket{00~}_{AB}\ket{\Phi^{+}}_{ab}~ &\longrightarrow_Q ~ \frac{1}{2}\left[\ket{00}_{Aa}\ket{00}_{Bb}+\ket{00}_{Aa}\ket{01}_{Bb}+i\ket{01}_{Aa}\ket{10}_{Bb}-i\ket{01}_{Aa}\ket{11}_{Bb}\right],\\
\ket{01~}_{AB}\ket{\Phi^{+}}_{ab}~ &\longrightarrow_Q ~ \frac{1}{2}\left[\ket{00}_{Aa}\ket{10}_{Bb}+\ket{00}_{Aa}\ket{11}_{Bb}-i\ket{01}_{Aa}\ket{00}_{Bb}+i\ket{01}_{Aa}\ket{01}_{Bb}\right],\\
\ket{1+_{\frac{\pi}{4}}}_{AB}\ket{\Phi^{+}}_{ab}~ &\longrightarrow_Q ~ \frac{1}{2}\left[\ket{10}_{Aa}\ket{00}_{Bb}+\ket{10}_{Aa}\ket{11}_{Bb}+i\ket{11}_{Aa}\ket{10}_{Bb}+i\ket{11}_{Aa}\ket{01}_{Bb}\right],\\
\ket{1-_{\frac{\pi}{4}}}_{AB}\ket{\Phi^{+}}_{ab}~ &\longrightarrow_Q ~ \frac{1}{2}\left[\ket{10}_{Aa}\ket{10}_{Bb}+\ket{10}_{Aa}\ket{01}_{Bb}+i\ket{11}_{Aa}\ket{00}_{Bb}-i\ket{11}_{Aa}\ket{11}_{Bb}\right].
\end{split}\end{equation}
\subsection{ Proof of optimality for $k=1$ and $\alpha={\pi}/{4}$}
Assume that $\PP^{(1)}(\beta)$ with $\beta\neq 0,\frac{\pi}{4}$ exists. It steers ${\cal A}^{(1)}(\beta)$ to the set with elements
\begin{equation}
\widetilde{\cal C}^{(1)}(\beta)=\sum_{j=1}^N a_{ij}\ket{\omega_{ij}}_{Aa}\ket{\xi_{ij}}_{Bb},
\end{equation}
where $i=1,...,4$, $a_{ij}$ are complex amplitudes, and $\omega_{ij}$, $\xi_{ij}$ are binary strings which correspond to particular sets of results obtained by Alice and Bob, respectively. These strings obey the following relations:
\begin{enumerate}[(1)]
\item \label{cond1}
\begin{equation}
\braket{\omega_{1j}}{\omega_{2j}}=\braket{\omega_{3j}}{\omega_{4j}}=1,
\end{equation}
which follow from the fact that initial states of Alice are indistinguishable within the two pairs. An instantaneous operation cannot increase their distinguishability.
\item
On the other hand, the orthogonality of Alice's states across the pairs $1-2$ and $3-4$ must be preserved; hence

\begin{equation}
\braket{\omega_{1j}}{\omega_{3j}}=\braket{\omega_{1j}}{\omega_{4j}}=\braket{\omega_{2j}}{\omega_{3j}}=\braket{\omega_{2j}}{\omega_{4j}}=0.
\end{equation}

\item Condition (\ref{cond1}) implies that on Bob's side the corresponding pairs need to remain orthogonal, i.e.,
\begin{equation}
\braket{\xi_{1j}}{\xi_{2j}}=\braket{\xi_{3j}}{\xi_{4j}}=1,
\end{equation}
\item \label{cond4} We assume that the final result of the nonlocal measurement is a $4$-valued function of all output bits; that is, none of the bits $A,B,a,b$ is redundant. This will be taken
into account when analysing the general structure of possible unitaries.

\end{enumerate}

\noindent In the following derivation, we will make use of the corollary of the nonsignaling principle.
\begin{cor}\label{nsp}
The reduced density matrix on Bob's side, $\rho_{Bb}$, is invariant under local actions on systems $A$ and $a$, and vice versa.
\end{cor}
As a trivial illustration of the corollary $\ref{nsp}$ let us show that the required map cannot be accomplished for $\alpha=\pi/4$ without ancillary systems. In this case there will be only four output binary strings (up to permutations), i.e.,
\begin{equation}\begin{split}
\ket{00~}_{AB}~ &\longrightarrow ~ a_{11}\ket{0}_A\ket{0}_B,\\
\ket{01~}_{AB}~ &\longrightarrow ~ a_{21}\ket{0}_A\ket{1}_B,\\
\ket{1+}_{AB}~ &\longrightarrow ~ a_{31}\ket{1}_A\ket{0}_B,\\
\ket{1-}_{AB}~ &\longrightarrow ~ a_{41}\ket{1}_A\ket{1}_B,
\end{split}\end{equation}
where $|a_{ij}|=1$ due to normalization.
Now, assume that the initial state is $\ket{00}_{AB}$. After the protocol is executed the output state remains unchanged. The matrix element on Bob's side $\rho_{B}^{(22)}=0$. If, however, Alice flips her spin shortly before the protocol is executed, the state becomes $\ket{10}=1/\sqrt{2}(\ket{1+}+\ket{1-})$ and the output state is $\ket{1}_A\ket{+}_B$, with $\rho_{B}^{(22)}=1/4$. Hence, such an operation is forbidden~\footnote{We could arrive to the same conclusion based on the fact that the corresponding transformation on  Bob's side  cannot be performed by unitary on $B$. }.\\
\subsection{One ebit is necessary and sufficient for  $\alpha = {\pi}/{4}$}
Let us show that an ebit is necessary and sufficient. In this case, there are 16 output binary strings $\ket{\omega_{ij}}_{Aa}\ket{\xi_{ij}}_{Bb}$.
Taking into account the conditions (\ref{cond1})-(\ref{cond4}) the action of $V_{Aa}\otimes U_{Bb}$ can be represented (up to trivial permutations) as

\begin{equation}\label{statesplitting}\begin{split}
\ket{00~}_{AB}\ket{\Phi}_{ab}~ &\longrightarrow ~ a_{11}\ket{00}_{Aa}\ket{00}_{Bb}+a_{12}\ket{00}_{Aa}\ket{01}_{Bb}+a_{13}\ket{01}_{Aa}\ket{10}_{Bb}+a_{14}\ket{01}_{Aa}\ket{11}_{Bb}\\
\ket{01~}_{AB}\ket{\Phi}_{ab}~ &\longrightarrow ~ a_{21}\ket{00}_{Aa}\ket{10}_{Bb}+a_{22}\ket{00}_{Aa}\ket{11}_{Bb}+a_{23}\ket{01}_{Aa}\ket{00}_{Bb}+a_{24}\ket{01}_{Aa}\ket{01}_{Bb}\\
\ket{1+}_{AB}\ket{\Phi}_{ab}~ &\longrightarrow ~ a_{31}\ket{10}_{Aa}\ket{00}_{Bb}+a_{32}\ket{10}_{Aa}\ket{11}_{Bb}+a_{33}\ket{11}_{Aa}\ket{10}_{Bb}+a_{34}\ket{11}_{Aa}\ket{01}_{Bb}\\
\ket{1-}_{AB}\ket{\Phi}_{ab}~ &\longrightarrow ~ a_{41}\ket{10}_{Aa}\ket{10}_{Bb}+a_{42}\ket{10}_{Aa}\ket{01}_{Bb}+a_{43}\ket{11}_{Aa}\ket{00}_{Bb}+a_{44}\ket{11}_{Aa}\ket{11}_{Bb}.
\end{split}\end{equation}

\begin{prop}\label{prop1}
$|a_{ij}|=1/4$ $\forall ~i,j$.
\end{prop}
\begin{proof}
The proof rests on a straightforward application of Corollary \ref{nsp}. Let $a_{ij}=|a_{ij}|e^{i\theta_{ij}}$.
 Assume that Alice and Bob are given $\ket{00}_{AB}$. After they implement the protocol the reduced density matrix on Bob's side is
\begin{equation}\label{redBob}
\rho_{Bb}=\left(\begin{array}{cccc}|a_{11}|^2&|a_{11}||a_{12}^{}|e^{i(\theta_{11}-\theta_{12})}&0&0\\
|a_{11}^{}| |a_{12}|e^{-i(\theta_{11}-\theta_{12})}&|a_{12}|^2&0&0\\
0&0&|a_{13}|^2&|a_{14}^{}|| a_{13}^{}|e^{i(\theta_{13}-\theta_{14})}
\\0&0&|a_{14}| |a_{13}^{}|e^{-i(\theta_{13}-\theta_{14})}&|a_{14}|^2\end{array}\right).
\end{equation}
However, if Alice flips her qubit shortly before the protocol is implemented, then the initial state is
\begin{equation}
\ket{10}_{AB}=\ket{1}_A(\cos\alpha\ket{+_\alpha}_B+\sin\alpha\ket{-_\alpha}_B).
\end{equation}
For simplicity let us consider the element in the first row and the first column of the resulting density matrix $\rho '_{Bb}$ on Bob's side:
\begin{equation}
\rho_{Bb}^{'(11)}=\cos^2\alpha |a_{31}|^2+\sin^2\alpha |a_{43}|^2.
\end{equation}

\noindent By Corollary \eqref{nsp} it has to match the corresponding element in \eqref{redBob}, i.e. $\rho_{Bb}^{(11)} = |a_{11}|^2$. The same argument can be applied to three other starting states of the set, which gives rise to the following set of simultaneous equations:

\begin{equation}\begin{cases}
\begin{split}
|a_{11}|^2&=\cos^2\alpha|a_{31}|^2+\sin^2\alpha|a_{43}|^2,\\
|a_{23}|^2&=\sin^2\alpha|a_{31}|^2+\cos^2\alpha|a_{43}|^2,\\
|a_{31}|^2&=\cos^2\alpha|a_{31}|^2+\sin^2\alpha|a_{43}|^2,\\
|a_{43}|^2&=\sin^2\alpha|a_{31}|^2+\cos^2\alpha|a_{43}|^2,\\
\end{split}\end{cases}
\end{equation}
which yields

\begin{equation}
|a_{11}|=|a_{23}|=|a_{31}|=|a_{43}|.
\end{equation}

\noindent Similarly, one can show equality of the remaining four quadruples. Also, notice that $|a_{11}|^2=|a_{12}|^2=|a_{13}|^2=|a_{14}|^2$ and the final result follows.
\end{proof}

Now, we are ready to prove our main result of this section:
\begin{thm}\label{onequbitthm}
For $\alpha=\pi/4$ the necessary and sufficient resource is a maximally entangled state of $a$ and $b$, where $a$, $b$ are both qubits.
\end{thm}
\begin{proof}
Sufficient condition follows from existence of explicit protocols, which achieve that goal. For the necessary condition observe that each of the output state contains
one ebit. Indeed, by Proposition~\ref{prop1} for $i=1$ the state can be written as

\begin{equation}
\ket{0}_A\otimes \frac{1}{\sqrt{2}}\left( \ket{0}_a \frac{e^{i\theta_{11}}\ket{00}_{Bb}+e^{i\theta_{12}}\ket{01}_{Bb}}{\sqrt{2}}+\ket{1}_a \frac{e^{i\theta_{13}}\ket{10}_{Bb}+e^{i\theta_{14}}\ket{11}_{Bb}}{\sqrt{2}}\right),
\end{equation}
which clearly contains one ebit. As local unitary operations $V_{Aa}$, $U_{Bb}$ cannot change entanglement of the state we conclude that $\ket{\Phi}_{ab}$ must contain an ebit.
\end{proof}
\subsection{Two qubits are necessary and sufficient for $\alpha = \pi/8$}
Similarly to the proof of Theorem~\ref{onequbitthm} two qubits are sufficient because there exists an explicit protocol ${\cal Q}^{(2)}[\pi/8]$. It remains to prove the necessary condition.
\begin{prop}\label{prop2}
For $\alpha\neq \pi/4$ one ebit is not sufficient.
\end{prop}
\begin{proof}
Assume again that Alice and Bob are given $\ket{00}_{AB}$ and compare the reduced matrix on Bob's side with the case when Alice flips her spin. By Proposition~\ref{prop1}
\begin{equation}
\rho_{Bb}=\frac{1}{4}\left(\begin{array}{cccc}1&e^{i(\theta_{11}-\theta_{12})}&0&0\\
e^{-i(\theta_{11}-\theta_{12})}&1&0&0\\
0&0&1&e^{i(\theta_{13}-\theta_{14})}
\\0&0&e^{-i(\theta_{13}-\theta_{14})}&1\end{array}\right).
\end{equation}
Now, in case Alice flips her qubit consider an off-diagonal element

\begin{equation}
\rho'^{(12)}_{Bb}=\frac{1}{8}\sin2\alpha(e^{i(\theta_{31}-\theta_{42})}+e^{i(\theta_{43}-\theta_{34})})
\end{equation}
of the new density matrix $\rho'_{Bb}$. By Corollary~\ref{nsp} it must be equal to $\rho^{(12)}_{Bb}$, i.e.

\begin{equation}
2e^{i(\theta_{11}-\theta_{12})}=\sin2\alpha(e^{i(\theta_{31}-\theta_{42})}+e^{i(\theta_{43}-\theta_{34})}).
\end{equation}
Equating real and imaginary components yields

\begin{equation}
\begin{cases}
\begin{split}
2\cos(\theta_{11}-\theta_{12})&=\sin 2\alpha \left(\cos(\theta_{31}-\theta_{42})+\cos(\theta_{43}-\theta_{34})\right),\\
2\sin(\theta_{11}-\theta_{12})&=\sin 2\alpha \left(\sin(\theta_{31}-\theta_{42})+\sin(\theta_{43}-\theta_{34})\right).
\end{split}
\end{cases}
\end{equation}
After adding the squared equations we obtain

\begin{equation}
2=\sin^2 2\alpha \left(1+\cos(\theta_{31}-\theta_{42}-\theta_{43}+\theta_{34})\right).
\end{equation}

The maximum of the RHS is $2\sin^2 2\alpha$, therefore the equality can be only achieved for $\alpha=\pi/4$.
\end{proof}

Although we believe that the above proof technique could be used to prove optimality of $\Q^{(k)}[\pi/2^{k+1}]$ for any value of $k>2$, it becomes computationally intractable.

\subsection{Proof for arbitrary $k$}
\begin{thm}
For $\alpha$ such that $|\bar\alpha|_2 = k$ consider the protocol $\Q^{(k)}[\alpha]$ (described in the earlier sections), which steers ${\cal A}^{(k)}[\alpha]$ to ${\cal C}^{(k)}[\alpha]$.  Then $\Q^{(k)}[\alpha]$ is an optimal protocol. 
\end{thm}
\begin{proof}
We prove by induction on $k$. We have shown earlier that the statement is true for $k=1,2$. Now, let us assume that $\Q^{(k-1)}[2\alpha]$ is optimal and prove that $\Q^{(k)}[\alpha]$ is optimal too. We prove by contradiction.
Assume that there exists a more economical protocol, $\PP^{(k-1)}[\alpha]$, which steers ${\cal A}^{(k-1)}[\alpha]$ to $\tilde{\cal C}^{(k-1)}[\alpha]$. The most general form of $\PP$ is
\begin{equation}\label{gen_prot}\begin{split}
\PP^{(m)}[\alpha]&= V_{A\tilde{a}_m}\otimes U_{B\tilde{b}_m}=\left(\ketbra{0}{0}_A\otimes V_{0,\tilde{a}_m}+\ketbra{1}{1}_A\otimes V_{1,\tilde{a}_m}\right )\otimes U_{B\tilde{b}_m}.
\end{split}
\end{equation}
where $\tilde{a}_m\equiv a_1\ldots a_m$, $\tilde{b}_m\equiv b_1\ldots b_m$ \footnote{In principle, $V_{A\tilde{a}_m}$ may be decomposed into a product of two unitaries: some general unitary $\tilde V_{A\tilde{a}_m}$ and a permutation unitary $T_{A\tilde{a}_m}$. The latter always exist for any $\tilde V_{A\tilde{a}_m}$ as it preserves orthogonality of bit strings in blocks ${\cal C}_0$ and ${\cal C}_1$. Thus, every unitary acting on $A\tilde{a}_m$ can be cast in the form \eqref{gen_prot}.}. The specific form of $V_{A\tilde{a}_m}$ follows from the fact that it is determined by its action on the two orthogonal subspaces of the $2^{k}$-dimensional Hilbert space. In general, $V_{0,\tilde{a}_m}$, $V_{1,\tilde{a}_m}$ and $U_{B\tilde{b}_m}$ depend on $\alpha$, while the known optimal protocol for ${\cal A}^{(k-1)}[2\alpha]$ which utilises $k-1$ ebits, $\Q^{(k-1)}[2\alpha]$, has $V_{{0,\tilde{a}_{k-1}}}=\mathbb{I}_{\tilde{a}_{k-1}}$, $V_{1,\tilde{a}_{k-1}}=\bigotimes_{i=1}^{k-1}R^x_{a_i}(2^{i}\alpha)$, $U_{B\tilde{b}_{k-1}}=\left[\prod_{i=2}^{k-1}\left((1-P_1^{(i-1)})\otimes \mathbb{I}_{Bb_i}+P_1^{(i-1)}\otimes\Omega_{Bb_i}\right)\right ]\Omega_{Bb_1}\equiv \Omega_{B\tilde{b}_{k-1}}$.\\

\noindent {\it Step 1: we reduce \eqref{gen_prot} to a special form \eqref{gen_prot1} below.}\\
\noindent Let us analyse how $\PP^{(k-1)}[\alpha]$ acts on ${\cal A}_0^{(k-1)}[\alpha]$ (recall that it is independent of $\alpha$). It is helpful to rewrite \eqref{gen_prot} as

\begin{equation}\label{gen_prot2}\begin{split}
\PP^{(k-1)}[\alpha]&= \left(\mathbb{I}_A\otimes V_{0,\tilde{a}_{k-1}}\otimes U_{B\tilde{b}_{k-1}}\Omega_{B\tilde{b}_{k-1}}^{\dagger}\right)\left( \ketbra{0}{0}_A\otimes\mathbb{I}_{\tilde{a}_{k-1}}\otimes \Omega_{B\tilde{b}_{k-1}}+\ketbra{1}{1}_A\otimes V_{0,\tilde{a}_{k-1}}^{-1} V_{1,\tilde{a}_{k-1}}\otimes \Omega_{B\tilde{b}_{k-1}} \right).
\end{split}
\end{equation}

\noindent Now
\begin{equation}
\left(\ketbra{0}{0}_A\otimes\mathbb{I}_{\tilde{a}_{k-1}}\otimes \Omega_{B\tilde{b}_{k-1}}\right){\cal A}_0^{(k-1)}[\alpha]={\cal C}_0^{(k-1)}[2\alpha],
\end{equation}
where ${\cal C}_0^{(k-1)}[2\alpha]$ is the set achieved by the standard protocol $\Q^{(k-1)}[2\alpha]$ implemented on ${\cal A}_0^{(k-1)}[2\alpha]$. This is true because ${\cal A}_0^{(k-1)}[\alpha]={\cal A}_0^{(k-1)}[2\alpha]$. But we also know from \eqref{gen_prot2} that
\begin{equation}
\PP^{(k-1)}[\alpha]{\cal A}_0^{(k-1)}[\alpha]=\tilde{\cal C}_0^{(k-1)}[\alpha]\neq{\cal C}_0^{(k-1)}[2\alpha],
\end{equation}
due to action of the first term in the brackets on the RHS, where

\begin{equation}
\tilde{\cal C}_0^{(k-1)}[\alpha]=\left(\mathbb{I}_A\otimes V_{0,\tilde{a}_{k-1}}\otimes U_{B\tilde{b}_{k-1}}\Omega_{B\tilde{b}_{k-1}}^{-1}\right){\cal C}_0^{(k-1)}[2\alpha].
\end{equation}
Also, $\tilde{\cal C}^{(k-1)}[\alpha]$ is still a set which satisfies standard conditions. Hence we observe that $\mathbb{I}_A\otimes V_{0,\tilde{a}_{k-1}}\otimes U_{B\tilde{b}_{k-1}}\Omega_{B\tilde{b}_{k-1}}^{-1}$ only reshuffles the
product basis states and, perhaps, introduces phases. We conclude that, there exists a protocol which acts on ${\cal A}_0^{(k-1)}[\alpha]$ as the standard optimal protocol $\Q^{(k-1)}[2\alpha]$.   In other words, the hypothetical protocol $\PP^{(k-1)}[\alpha]$ can be assumed to be acting on ${\cal A}_0^{(k-1)}[\alpha]$ in the same way as $\Q^{(k-1)}[2\alpha]$ does on ${\cal A}_0^{(k-1)}[2\alpha]$. in other words, we can assume that

\begin{equation}\label{gen_prot1}\begin{split}
\PP^{(k-1)}[\alpha]&= \left( \ketbra{0}{0}_A\otimes\mathbb{I}_{\tilde{a}_{k-1}}+\ketbra{1}{1}_A\otimes V_{\tilde{a}_{k-1}}\right)\otimes \Omega_{B\tilde{b}_{k-1}},
\end{split}
\end{equation}
where $V_{\tilde{a}_{k-1}}\equiv V_{0,\tilde{a}_{k-1}}^{-1} V_{1,\tilde{a}_{k-1}}$. It would have been identical to $\Q^{(k-1)}[2\alpha]$ if $V_{0,\tilde{a}_{k-1}}^{-1} V_{1,\tilde{a}_{k-1}}=\bigotimes_{i=1}^{k-1}R^x_{a_i}(2^{i}\alpha)$. Clearly, this cannot be the case.\\

\noindent {\it Step 2: We extend the protocol to make its action comparable to that of $\Q^{(k)}[\alpha]$.}\\
\noindent Let us return to the original protocol $\Q^{(k)}[\alpha]$, which gives ${\cal C}^{(k)}[\alpha]={\cal C}_0^{(k)}\oplus{\cal C}_1^{(k)}[\alpha]$. Now consider extension of $\PP^{(k-1)}[\alpha]$, with an additional ebit
as follows:
\begin{equation}
\mathbb{I}_A\otimes\mathbb{I}_{a_k}\otimes\left((1-P_1^{(k-1)})\otimes \mathbb{I}_{Bb_k}+P_1^{(k-1)}\otimes\Omega_{Bb_k}\right)\PP^{(k-1)}[\alpha]\equiv \tilde{\PP}^{(k)}[\alpha].
\end{equation}

\noindent This extended protocol, albeit consuming a reduntant ebit, still accomplishes the task, i.e. 

 \begin{equation}
\tilde{\PP}^{(k)}[\alpha]{\cal A}^{(k)}[\alpha]={\tilde{\cal C}}^{(k)}[\alpha]={\cal C}_0^{(k)}\oplus\tilde{\cal C}_1^{(k)}[\alpha],
\end{equation}
and can now be compared with the action of $\Q^{(k)}[\alpha]$,
\begin{equation}
\Q^{(k)}[\alpha]{\cal A}^{(k)}[\alpha]={\cal C}^{(k)}[\alpha]={\cal C}_0^{(k)}\oplus{\cal C}_1^{(k)}[\alpha],
\end{equation}

\noindent where ${\cal C}^{(k)}[\alpha]$ and $\tilde{\cal C}^{(k)}[\alpha]$ are locally equivalent by Proposition \ref{master_set} hence there exists a permutation of the basis vectors which reduces $\cal Q$ to $\tilde {\cal P}$ .\\

\noindent {\it Step 3: We show that $\cal Q$ is not reducible to $\tilde {\cal P}$.}\\
We can represent $\cal Q$ and $\tilde {\cal P}$ as follows:

\begin{align}
\Q^{(k)}[\alpha]&= \left( \ketbra{0}{0}_A\otimes\mathbb{I}_{\tilde a_k}+\ketbra{1}{1}_A\otimes_{i=1}^k R^x_{a_i}[2^{i}\alpha]\right)\otimes \Omega_{Bb}\\
\tilde\PP^{(k)}[\alpha]&= \left( \ketbra{0}{0}_A\otimes\mathbb{I}_{\tilde a_k}+\ketbra{1}{1}_A\otimes V_{\tilde a_{k-1}}\otimes\mathbb{I}_{a_k}\right)\otimes \Omega_{Bb}.
\end{align}

Consider the ancillary subspace ${\cal H}_{|1\rangle\langle 1|}\otimes{\cal H}_{\tilde a_k}$. On this subspace $\cal Q$ acts as a tensor product of $k$ single-qubit rotations $\otimes_{i=1}^{k} {R}(2^{i}\alpha)$, and $\cal P$ acts as $ V_{\tilde a_{k-1}}\otimes \mathbb{I}_{a_k}$. There exists $\alpha$ for which the spectrum of the total rotation consists of $2^k$ different eigenvalues. The spectrum of the unitary $V_{\tilde a_{k-1}}\otimes\mathbb{I}_{a_k}$ consists of at most $2^{k-1}$ different eigenvalues. Thus, the latter protocol cannot be reduced to the former which contradicts our initial assumption.
%
%
\end{proof}

\section{A (Possible) Generalization of the steering sets}

Here we give two examples of sets of entangled states with distinguishability properties equivalent to the sets of product states studied earlier.
\begin{examp}
The following set of two orthogonal entangled states is equivalent to ${\cal A}[\alpha_1,...,\alpha_n]$ \footnote{It is known that any two pure orthogonal states, entangled or product, can be perfectly distinguished by LOCC~\cite{walgate_local_2000}}.

\begin{equation}\label{ent_set}
 \begin{split}
  \ket{\psi}&=c_0\ket{0}_A\ket{0}_B+c_1\ket{1}_A\ket{+_{\alpha_1}}_B+...+c_n\ket{n}_A\ket{+_{\alpha_n}}_B\\
\ket{\phi}&=d_0\ket{0}_A\ket{1}_B+d_1\ket{1}_A\ket{-_{\alpha_1}}_B+...+d_n\ket{n}_A\ket{-_{\alpha_n}}_B
 \end{split}
\end{equation}
\end{examp}

\noindent Clearly, a local projection measurement in the basis $\{\ket{i}_A\}_{i=1}^{n}$ by Alice maps this set to ${\cal A}[\alpha_1,...,\alpha_n]$.
Hence, its distinguishability is necessary and
sufficient for distinguishability of~\eqref{set1}.

\begin{examp}
The set of four entangled states
\be\label{belllike}{\cal D}[\alpha]:=\{|\Psi_i\ra\}_{i=1}^4,
\ee
 where
\begin{align*}
&|\Psi_1\ra = \cos\alpha|00\ra+\sin\alpha|11\ra,\\
&|\Psi_2\ra=\sin\alpha|00\ra-\cos\alpha|11\ra,\\
&|\Psi_3\ra=\cos\alpha|01\ra+\sin\alpha|10\ra,\\
&|\Psi_4\ra=\sin\alpha|01\ra-\cos\alpha|10\ra.
\end{align*}
\end{examp}

Let us show that ${\cal D}\rightarrow \cal B$ requires the same amount of entanglement resources as $\cal A\rightarrow \cal B$. First stage is to steer ${\cal D}$ to ${\cal A}[2\alpha]$. The protocol which achieves such a map is a modification of the procedure invented in \cite{IMNLV_GR}. Alice and Bob start with implementing instantaneous nonlocal C-NOT
with $B$ and $A$ being a control and a target respectively (as a part of this step they obtain corresponding results of local measurements -- $x_a=\pm 1$ and
$z_b=\pm 1$). This is followed by Bob rotating his qubit $B$ by $\alpha$ about the $y$ axis.
The resulting mapping is summarised in Table \ref{table:D}.

\begin{table}[h!]
\begin{center}
\begin{tabular}{cclcl}\hline\label{ab}
$x_a \backslash
z_b$&~&$z_b=+1$&~&$z_b=-1$\\
\hline\hline\noalign{\smallskip}
$x_a=+1$&~&$\ket{0}_A\ket{+_{2\alpha}}_B$&~&$\ket{1}_A\ket{+_{2\alpha}}_B$\\
&~&$\ket{0}_A\ket{-_{2\alpha}}_B$&~&$\ket{1}_A\ket{-_{2\alpha}}_B$\\
&~&$\ket{1}_A\ket{+_{2\alpha}}_B$&~&$\ket{0}_A\ket{+_{2\alpha}}_B$\\
&~&$\ket{1}_A\ket{-_{2\alpha}}_B$&~&$\ket{0}_A\ket{-_{2\alpha}}_B$\\
\noalign{\smallskip}\hline\noalign{\smallskip}
$x_a=-1$&~&$\ket{0}_A\ket{0}_B$&~&$\ket{1}_A\ket{0}_B$\\
&~&$\ket{0}_A\ket{1}_B$&~&$\ket{1}_A\ket{1}_B$\\
&~&$\ket{1}_A\ket{0}_B$&~&$\ket{0}_A\ket{0}_B$\\
&~&$\ket{1}_A\ket{1}_B$&~&$\ket{0}_A\ket{1}_B$\\
\noalign{\smallskip}\hline

\end{tabular}
\end{center}
\caption{The resulting map after first stage.}\label{table:D}
\end{table}

Here the states of $A$ and $b$ are in direct product with the set
\begin{equation}\begin{split}
&\ket{-}_a\ket{0}_B\\
&\ket{-}_a\ket{1}_B\\
&\ket{+}_a\ket{+_{2\alpha}}_B\\
&\ket{+}_a\ket{-_{2\alpha}}_B,
\end{split}
\end{equation}
which is equivalent to ${\cal A}[2\alpha]$ and requires ${\cal I}(\alpha)-1$ ebits to steer it to $\cal B$ on the second stage, bringing the total cost to ${\cal I}(\alpha)$.

\section{discussion}
We have considered the question of distinguishability of sets of mutually orthogonal states under LOSE. Our main objects of interest were specific families of sets of bipartite product states ~\eqref{set1} characterised by $n$ real parameters $\{\alpha_i\}_{i=1}^n$, with ~\eqref{twisted_product} being a simplest case of a single parameter $\alpha$. The task had been represented as a steering game and we have provided an algorithm which achieves the final state of the game, i.e. a set composed of tensor-product local bases states, using optimal amount of entanglement. For the steering game we defined a figure of merit  $\cal I$ which quantifies the exact amount of entanglement required to steer the initial set to the given one and hence can be regarded as a measure of information nonlocality in the set.
We have shown that it is a maximal length of binary expansions of $\alpha_i$.

We have provided examples of sets of entangled states, which can be steered using the same number of ebits as the sets of product states with the
same values of corresponding parameters. In order to show that the type of nonlocality that we study is not associated with entanglement of a set, but depends solely on the values of parameters characterising them, let us consider particular mixtures of the states.
\begin{examp} Alice and Bob are given a state drawn from ${\cal D}[\alpha]$ ~\eqref{belllike} with a probability distribution such that the resulting mixed state is
\be
\rho_W = \frac{1-F}{3}\left(\Psi_1+\Psi_2+\Psi_3\right)+F\Psi_4=\frac{1-F}{3}{\mathbb{I}}+ \frac{4F-1}{3}\Psi_4,
\ee
where $0\leq F\leq 1$. It is known that for $\alpha=\pi/4$ the mixed state $\rho_W$ is not entangled when $0\le F\le\half$ and entangled when $\half < F \le 1$ ~\cite{werner_quantum_1989} with $\Psi = \ketbra{\Psi}{\Psi}$ and $\Psi_1 = \Phi^+, \Psi_2 = \Phi^-, \Psi_3 = \Psi^+, \Psi_4 = \Psi^-$ -- the four Bell states. However, as was shown earlier, the steering properties of $\rho_W$ depends only on the value of $\alpha$, but not on $F$.
\end{examp}

Let us discuss our results in relation to past works. The question of designing more efficient protocols had been addressed in ~\cite{clark_entanglement_2010}. The distinguishing feature of their proposal is that the protocol halts with increasing probability as the number of rounds grows larger. Thereby, it does not consume the remaining entanglement resources.  However, there is non-zero probability that the protocol will consume all available entanglement.  Thus, the protocol had been designed to minimise the expected number of ebits consumed. Our goal is different: we aim at minimizing the amount of entanglement that has to be available beforehand.

Two correlation measures related to our scenario --- {\it quantum deficit} and {\it classical zero-way deficit} --- have been introduced in ~\cite{horodecki_local_2005}.

We note that they are unsuitable measures of the nonlocality of ${\cal A}[\alpha]$: quantum deficit may be made large or small depending on $\alpha$, and independent of the binary expansion of the latter. Classical zero-way deficit depends on the probability distribution of the states drawn from $\cal A[\alpha]$
in~\eqref{twisted_product}, whereas the amount  of nonlocality of the given set as quantified by a steering task is not --- it depends only on $\alpha$.

Our results are closely related to the definition of quantumness of sets of quantum states proposed in~\cite{groisman_quantumness_2007}. The sets considered in our work fall into the category of nonclassical (quantum) sets, and the presence of de-localised information can be interpreted as a signature of quantumness. However, the measure of information nonlocality proposed by us is profoundly different from the {\it relative entropy of quantumness} proposed in ~\cite{groisman_quantumness_2007}. The latter, being a distance-measure, depends on the probability distribution of the members of the set. It is important to note that information nonlocality captures the information-theoretic aspect of sets in terms of the length of binary representation of corresponding paremeter(s) and as such cannot serve as an alternative measure of quantumness in a distance sense --- two very close sets can have very different values of information nonlocality.

Our work raises a number of questions. First, what are the most general orthogonal sets which one can steer to $\cal B$ and how does our measure $\cal I$ generalise to these sets? Second, what are the applications of the steering algorithm to other tasks of quantum information processing? Third, whether one can relate the number of ebits needed to bi-localise information to the cost of resetting delocalised information in the sense of Landauer's principle. Finally, what is the meaningful generalization of the steering procedure to the multipartite states for three parties and more?
\section*{Acknowledgments}
The authors would like to thank Sidney Sussex College, Cambridge for financial support.

\appendix*
\section{the instantaneous rotations protocol -- the stator formalism}\label{appendix}
The instantaneous steering protocol of Sec.~\ref{genstsect} can be presented in the state-operator (stator) formalism which was first introduced in~\cite{reznik_remote_2002}. Steps 1 and 2 result in the preparation of the following stator:
\be
\ket{0}_{b_t}\otimes S_+^{(t)} + \ket{1}_{b_t}\otimes S_-^{(t)},
\ee
where $S_{\pm}^{(t)}=\ket{0}_{a_t}\otimes\mathbb{I}_B\pm\ket{1}_{a_t}\otimes\sigma_y^B$ and we omit normalization. Both $S_{\pm}^{(t)}$ satisfy eigen-operator equation $\sigma_x^{a}S_{\pm}^{(t)}=\pm \sigma_{y}^B S_{\pm}^{(t)}$; thus 
\be
\sigma_x^{a_t}\left(\ket{0}_{b_t}\otimes S_+^{(t)} + \ket{1}_{b_t}\otimes S_-^{(t)}\right)=\sigma_y^{B}\left(\ket{0}_{b_t}\otimes S_+^{(t)} - \ket{1}_{b_t}\otimes S_-^{(t)}\right),
\ee
which in turn implies that 
\be
R_a(\theta)\left(\ket{0}_{b_t}\otimes S_+^{(t)} + \ket{1}_{b_t}\otimes S_-^{(t)}\right)=\ket{0}_{b_t}\otimes R_B(\theta)S_+^{(t)}  + \ket{1}_{b_t}\otimes R_B(-\theta) S_-^{(t)}.
\ee
This ability to ``propagate" transformations from $a$ to $B$ lies in the heart of stator formalism.

The sequence of Bob's actions on $B$ and $k$ ebits $\ket{\Phi^+}^{\otimes k}$ [the unitary on the RHS of Eq.~\eqref{instrot}] results in what we call the ``superstator":
\be
S_{sup}=\sum_{t=1}^k \frac{1}{2^{t/2}} \ket{\Phi^+}_{a_kb_k}...\ket{\Phi^+}_{a_{t+1}b_{t+1}}\ket{0_{b_t} 1_{b_{t-1}}...1_{b_1}}S_+^{(t)}S_-^{(t-1)}...S_-^{(1)} + \Theta,
\ee
where $\Theta = 2^{-k/2}\ket{1_{b_k} 1_{b_{k-1}}...1_{b_1}}S_-^{(k)}S_-^{(k-1)}...S_-^{(1)}$.
For each term $S_{sup}$ the actions on Alice's side generate a desired rotation of $B$:
\begin{align}
&\left[\bigotimes_{i=1}^t R^x_{a_i}(2^{i-1}\alpha)\right]S_+^{(t)}S_-^{(t-1)}...S_-^{(1)}=R_B^y(\alpha)S_+^{(t)}S_-^{(t-1)}...S_-^{(1)},\\
&\left[\bigotimes_{i=1}^k R^x_{a_i}(2^{i-1}\alpha)\right]S_-^{(k)}S_-^{(k-1)}...S_-^{(1)}=R_B^y(-\pi/2+\alpha)S_-^{(k)}S_-^{(k-1)}...S_-^{(1)}.
\end{align}

\bibliographystyle{apsrev4-1}

%

\end{document}